\documentclass[11pt]{article}
\usepackage{amsmath,amssymb,amsfonts,amsthm,latexsym,epsfig,hyperref,array,enumerate,geometry,graphicx,url,tikz}
\usepackage[all]{xy}
\newtheorem{theorem}{\bf Theorem}
\newtheorem{lemma}[theorem]{\bf Lemma}
\newtheorem{proposition}[theorem]{\bf Proposition}
\newtheorem{corollary}[theorem]{\bf Corollary}
\newtheorem{remark}[theorem]{\bf Remark}
\newtheorem{definition}[theorem]{\bf Definition}
\newtheorem{example}[theorem]{\bf Example}
\numberwithin{equation}{section}

\usepackage [english]{babel}
\usepackage [autostyle, english = american]{csquotes}
\MakeOuterQuote{"}

\begin{document}
\title{On the linear structures of Balanced functions and quadratic APN functions}

\author{A. Musukwa and  M. Sala}

\date{}
\maketitle  
\begin{center}{University of Trento, Via Sommarive, 14, 38123 Povo, Trento, Italy\\{\{augustinemusukwa, maxsalacodes\}@gmail.com}}\end{center}            
\begin{abstract}
The set of linear structures of most known balanced Boolean functions is non-trivial. In this paper, some balanced Boolean functions whose set of linear structures is trivial are constructed. We show that any APN function in even dimension must have a component whose set of linear structures is trivial. We determine a general form for the number of bent components in quadratic APN functions in even dimension and some bounds on the number are produced. We also count bent components in any quadratic power functions.
\end{abstract}


\section{Introduction}
Balancedness is an important property which is sometimes required in Boolean functions since it is often desirable for cryptographic primitives to be unbiased in output. By recognising such importance, a lot of papers have been written on construction of balanced functions with cryptographic properties (see for example \cite{Chak,Cus1,Kho,Seb}). One cryptographic property which is mostly considered in constructing such functions is nonlinearity. However, in this study we are interested in something different. We would like to consider the set of linear structures of balanced functions. In this paper, the set of linear structures of a Boolean function is called linear space. We believe that most known balanced functions do have non-trivial linear space. A typical example of a balanced function with non-trivial linear space is $g(x_1,...,x_{n-1})+x_n$, where $n$ is positive integer. It is a well-known balanced function and its linear space clearly includes the nonzero vector $(0,...,0,1)$. In this paper we construct some balanced functions whose linear spaces are trivial and in some cases we give a lower bound on their nonlinearities. 

The nonlinearity and differential uniformity of a vectorial Boolean function (a mapping from $\mathbb{F}_2^n$ to $\mathbb{F}_2^n$)  are properties which are used to measure the resistance of a function towards linear and differential attacks, respectively. APN and AB functions provide optimal resistance against the said attacks. This gives a justification as to why there are many studies regarding APN and also AB functions. In this paper, we show that the linear spaces of some components of an APN function in even dimension must be trivial. In particular, we show that the dimension of the linear space of any component in APN permutation is at most 1. Some results on properties of quadratic APN functions are studied. It is well-known that in any quadratic APN functions there are some bent components. So we provide a general form of the number of bent components in quadratic APN. From \cite{Pott}, we know that there are at most $2^n-2^{n/2}$ bent components in any function from $\mathbb{F}_2^n$ to itself. This motivated the authors to count bent components for any quadratic power function and so a comparison with quadratic power APN is made. 

This paper is organised as follows. In Section \ref{preliminaries}, known results are reported. In Section \ref{balanced-section}, some balanced functions are constructed and in Section \ref{linear structures-section} we provide conditions which help to determine whether the balanced functions constructed in Section \ref{balanced-section} have trivial linear space. In Section \ref{APN functions-section}, we show that there is component in any APN function in even dimension whose linear space is trivial and we also present a general form for the number of bent components in any quadratic APN functions. In Section \ref{power functions-section}, we count bent components in any quadratic power functions.  

\section{Preliminaries}\label{preliminaries}
In this section, some definitions and well-known results are reported  and for details, the reader is referred to \cite{Ber,Car1,Chee,Mac,Wu}.

The field of two elements, $0$ and $1$, is denoted by $\mathbb{F}$.  A vector in the vector space $\mathbb{F}^n$ is denoted by $v$. A vector whose $i$th coordinate is $1$ and $0$ elsewhere is denoted by $e_i$. We use ordinary addition $+$ instead of XOR $\oplus$. For any set $A$, its size is denoted by $|A|$.

A {\em Boolean function (B.f.)} is any function $f$ from $\mathbb{F}^n$ to $\mathbb{F}$ and a {\em vectorial Boolean function (v.B.f.)} is any function $F$ from $\mathbb{F}^n$ to $\mathbb{F}^m$, with $n,m\in\mathbb{N}$. However, in this paper we consider v.B.f.'s from $\mathbb{F}^n$ to $\mathbb{F}^n$. The B.f.'s in algebraic normal form, which is the $n$-variable polynomial representation over $\mathbb{F}$, is given by \(f(x_1,...,x_n)=\sum_{I\subseteq \mathcal{P}}a_I\left(\prod_{i\in I}x_i\right),\) where $\mathcal{P}=\{1,...,n\}$ and $a_I\in \mathbb{F}$. The {\em algebraic degree} or simply {\em degree} of $f$ (denoted by $\deg(f)$) is  $\max_{I\subseteq \mathcal{P}}\{|I|\mid a_I\ne 0\}$. The set of all B.f.'s on $n$ variables is denoted by $B_n$.

A B.f. $f$ is {\em linear} if $\deg(f)=1$ and $f(0)=0$, {\em affine} if $\deg(f)\le 1$, {\em quadratic} if $\deg(f)=2$ and {\em cubic} if $\deg(f)=3$. The set of all affine functions is denoted by $A_n$. Given a v.B.f. $F=(f_1,...,f_n)$, the functions  $f_1,...,f_n$ are called {\em coordinate functions} and the functions $\lambda\cdot F$, with $\lambda\neq 0\in \mathbb{F}^n$ and "$\cdot$" denoting dot product, are called a {\em components} of $F$ and we denote $\lambda\cdot F$ by $F_\lambda$. A v.B.f. $F$ is said to be a {\em permutation} if and only if all its components are balanced. The degree of a v.B.f. $F$ is given by $\deg(F)=\max_{\lambda\neq 0\in \mathbb{F}^n}\{\deg(F_\lambda)\}$. If all components of a v.B.f. $F$ are quadratic, we say that $F$ is {\em pure quadratic}. 

 The {\em Hamming weight} of $f$ is given by $\mathrm{w}(f)=|\{x\in \mathbb{F}^n\mid f(x)=1\}|$. A function $f$ is {\em balanced} if $\mathrm{w}(f)=2^{n-1}$. The {\em distance} between $f$ and $g$ is $d(f,g)=\mathrm{w}(f+g)$ and the {\em nonlinearity} of $f$  is $\mathcal{N}(f)=\min_{\omega \in A_n}d(f,\omega)$.
 
 For $m<n$, if $f$ is in $B_n$ but depends only on $m$ variables, then its restriction to these $m$ variables is denoted by $f_{\restriction\mathbb{F}^m}$. Clearly, $f_{\restriction\mathbb{F}^m}\in B_m$. 
 
The next result can be found in \cite{Mac} on page 372. 
 \begin{proposition}\label{balanced-splitting}
 	If $g(x_1,...,x_{n-1})$ is an arbitrary B.f.on $n-1$ variables, with a positive integer $n>1$, then \(f=g(x_1,...,x_{n-1})+x_n\) is balanced.
 \end{proposition}

The {\em Walsh transform} of a B.f. $f$ is defined as the function $W_f$ from $\mathbb{F}^n$ to $\mathbb{Z}$: \[W_f(a)=\sum_{x\in\mathbb{F}^n}(-1)^{f(x)+a\cdot x}\,,\] for all $a \in \mathbb{F}^n$. The set $\{W_f(a)\mid a\in\mathbb{F}^n\}$ is called the {\em Walsh spectrum} of a B.f. $f$. The {\em Walsh spectrum} of a v.B.f $F$ is given by $\{W_{F_\lambda}(a)\mid a\in\mathbb{F}^n,\lambda\neq 0\in\mathbb{F}^n\}$.
Define $\mathcal{F}(f)$ as
\[\mathcal{F}(f)=W_f(0)=\sum_{x\in \mathbb{F}^n}(-1)^{f(x)}=2^n-2\mathrm{w}(f).\]
Note that $f$ is balanced if and only if $\mathcal{F}(f)=0$.

The nonlinearity of a function $f$ can be written in terms of Walsh transform as \[\mathcal{N}(f)=2^{n-1}-\frac{1}{2}\max\limits_{a\in \mathbb{F}^n}|W_f(a)|.\] The nonlinearity of a v.B.f $F$ is defined as \[\mathcal{N}(F)=\min_{\lambda\neq 0\in\mathbb{F}^n}\mathcal{N}(F_\lambda).\]  It well-known that for every B.f. $f\in B_n$, with $n$ even, \(\mathcal{N}(f)\leq 2^{n-1}-2^{\frac{n}{2}-1}.\) A function $f\in B_n$ is said to be {\em bent} if $\mathcal{N}(f)=2^{n-1}-2^{\frac{n}{2}-1}$ and this can happen only in even dimension. Note that the lowest possible value for $W_f(a)$, with $a\in\mathbb{F}^n$, is $2^{\frac{n}{2}}$ and this bound is achieved only for bent functions. 

For $n$ odd, a B.f. $f$ is called {\em semi-bent} if $\mathcal{N}(f)=2^{n-1}-2^{\frac{n-1}{2}}$. In other words, $f$ is semi-bent if, for all $a\in\mathbb{F}^n$, $W_f(a)\in\{0,\pm 2^{\frac{n+1}{2}}\}$. Semi-bent functions are sometimes defined in even dimension. For $n$ even, we say a function $f$ is semi-bent if, for all $a\in\mathbb{F}^n$, $W_f(a)\in\{0,\pm 2^{\frac{n+2}{2}}\}$. A v.B.f. $F$ in odd dimension is {\em almost-bent (AB)} if all its components are  semi-bent.

A B.f. is called {\em plateaued} if its Walsh transform takes at most three values: $0$ and $\pm\mu$ where $\mu$ is some positive integer, called the {\em amplitude} of the plateaued function. So clearly bent and semi-bent functions are plateaued. 

We define the {\em (first-order) derivative} of $f$ at $a$ by $D_af(x)=f(x+a)+f(x)$. A derivative of $f$ at $0$ is the trivial derivative and at any other point, $a\neq 0\in\mathbb{F}^n$, we simply say a  derivative. An element $a\in\mathbb{F}^n$ is called a {\em linear structure} of $f$ if $D_af$ is constant, and we denote the set of all linear structures of $f$ by $V(f)$. The set $V(f)$ is called the {\em linear space} of $f$. We say the linear space  is {\em trivial} if it contains zero vector only and {\em non-trial} otherwise.

\begin{theorem}\label{bent-thm}
	A B.f. $f$ on $n$ variables is bent if and only if $D_af$ is balanced for any nonzero $a\in\mathbb{F}^n$.
\end{theorem}

Two B.f.'s $f,g:\mathbb{F}^n\rightarrow \mathbb{F}$ are said to be {\em affine equivalent} if there exist an affinity $\varphi:\mathbb{F}^n\rightarrow \mathbb{F}^n$ such that $f=g\circ \varphi$. This relation is denoted by $\sim_A $ and written as $f\sim_A g$. Observe that the relation $\sim_A$ is equivalence relation. For $i \in \{1,...,n\}$ and $l\in A_{n-1}$, a {\em basic affinity} of $\mathbb{F}^n$ maps $x_i\mapsto x_i+l(x_1,...,x_{i-1},x_{i+1},...,x_n)$ and fixes all other coordinates.  
	
\begin{proposition}\label{equivalence-properties}
	Let $f,g\in B_n$ be such that $f\sim_A g$. Then $\mathrm{w}(f)=\mathrm{w}(g)$ and so $f$ is balanced if and only if $g$ is balanced.
\end{proposition}

\begin{remark}\label{fourier-affine-equivalence-invariant}
	From Proposition \ref{equivalence-properties} and applying the fact that $\mathcal{F}(f)=2^n-2\mathrm{w}(f)$, we can easily deduce that if $f\sim_A g$ then $\mathcal{F}(f)=\mathcal{F}(g)$, that is, $\mathcal{F}(f)$ is invariant under affine equivalence.
\end{remark}

\begin{theorem}\label{quadratic}
	Let $f\in B_n$ be quadratic. Then
	
	\begin{itemize}
		\item[(i)] $f\sim_A x_1x_2+\cdots +x_{2k-1}x_{2k}+x_{2k+1}$, with $k\leq \lfloor \frac{n-1}{2}\rfloor$ if $f$ is balanced,
		\item[(ii)] $f\sim_A x_1x_2+\cdots +x_{2k-1}x_{2k}+c$,  with $k\leq \lfloor \frac{n}{2}\rfloor$ and $c\in\mathbb{F}$ if $f$ is unbalanced.
	\end{itemize}
\end{theorem} 

\begin{remark}\label{dimension-semi-bent}
	By Theorem \ref{quadratic}, it can be easily deduced that if $n$ is even then the dimension of the linear space of any quadratic function is even, and if $n$ is odd then its dimension is also odd.  
\end{remark}

The following corollary can be easily proved. 

\begin{corollary}\label{quadratic-bent-semi-bent}
	Let $f$ be a quadratic B.f. on $n$ variables and $c\in\mathbb{F}$. Then
	\begin{itemize}
		\item[1.]for even $n$, $f$ is bent if and only if 
		\begin{center}
			$f\sim_A x_1x_2+\cdots+x_{n-1}x_n+c$
		\end{center}
	
		\item[2.] for even $n$, $f$ is semi-bent if and only if $f\sim_A x_1x_2+\cdots+x_{n-3}x_{n-2}+x_{n-1}$ or $f\sim_A x_1x_2+\cdots+x_{n-3}x_{n-2}+c$.
		
		\item[3.] for odd $n$, $f$ is semi-bent if and only if $f\sim_A x_1x_2+\cdots+x_{n-2}x_{n-1}+x_n$ or $f\sim_A x_1x_2+\cdots+x_{n-2}x_{n-1}+c$.
	\end{itemize}
\end{corollary}

\begin{lemma}[\cite{Cus}]\label{weight-affine-equivalent-quadratics}
	Two unbalanced quadratic B.f. $g$ and $h$ on $n$ variables are affine equivalent if and only if $\mathrm{w}(g)=\mathrm{w}(h)$ and $\mathcal{N}(g)=\mathcal{N}(h)$.
\end{lemma}

\begin{theorem}[\cite{Car4}]\label{quadratic-nonlinearity}
	Let $f\in B_n$ be a quadratic function. Then, for $a\in \mathbb{F}^n$, we have $W_f(a)\in \{0,\pm 2^{\frac{n+k}{2}}\}$ and \(\mathcal{N}(f)=2^{n-1}-2^{\frac{n+k}{2}-1},\) where $k=\dim V(f)$.
\end{theorem} 

\begin{definition}
	Define \(\delta_F(a,b)=|\{x\in \mathbb{F}^n\mid D_aF(x)=b\}|\), for $a,b\in \mathbb{F}^n$ and v.B.f. $F$. The {\em differential uniformity of $F$ is} \(\delta(F)=\max_{a\neq 0,b\in \mathbb{F}^n}\delta_F(a,b)\) and always satisfies $\delta(F)\geq 2$. We call a function with $\delta(F)=2$ {\em Almost Perfect Nonlinear (APN)}.
\end{definition}

Next we look at another representation of v.B.f., known as {\em univariate polynomial representation}, which will be used in some sections. Consider the finite field $\mathbb{F}_{2^n}$ consisting of $2^n$ elements. It is well-known that the set $\mathbb{F}_{2^n}^*=\mathbb{F}_{2^n}\setminus\{0\}$ is a cyclic group which has $2^n-1$ elements. An element in $\mathbb{F}_{2^n}$ which is a generator of the multiplicative group $\mathbb{F}_{2^n}^*$ is called a {\em primitive element}. It is well explained in \cite{Car1} that the vector space $\mathbb{F}^n$ can be endowed with the structure of the finite field $\mathbb{F}_{2^n}$. So any function $F$ from $\mathbb{F}_{2^n}$ into $\mathbb{F}_{2^n}$ admits a unique univariate polynomial representation over $\mathbb{F}_{2^n}$, given as: 
\begin{align}
F(x)&=\sum_{i=0}^{2^n-1}\delta_ix^i,
\end{align} 
where $\delta_i\in\mathbb{F}_{2^n}$ and the degree of $F$ is at most $2^n-1$. Given the binary expansion $i=\sum_{s=0}^{n-1}i_s2^s$, define ${\rm w}_2(i)=\sum_{s=0}^{n-1}i_s$. So $F$ is a v.B.f. whose algebraic degree is given by \(\max\{{\rm w}_2(i)\mid 0\leq i\leq 2^n-1, \delta_i\neq 0\}\) (see \cite{Car1}).  

The (absolute) trace function $Tr:\mathbb{F}_{2^n}\rightarrow\mathbb{F}_2$ is defined as 
\begin{align*}
Tr(x)=x+x^2+x^{2^2}+\cdots+x^{2^{n-1}},
\end{align*} where $x\in\mathbb{F}_{2^n}$. For $\alpha\in\mathbb{F}_{2^n}$, a component $F_\alpha$ of $F$ is defined as $F_\alpha(x)=Tr(\alpha F)$.

We call any function of the form $F(x)=x^d$, for some non negative integer $d$, a {\em power function} and if $d=2^i+2^j$, for some non negative integers $i$ and $j$, with $i\neq j$, we say it is quadratic power function.

\section{Balanced Boolean functions}\label{balanced-section}
In this section we determine some conditions for B.f.'s to be balanced and we also construct some balanced functions.

If a B.f. is expressed in some particular form, its weight can be obtained from the weights of other B.f.'s on vector spaces with lower dimension as we explain below. First, observe that any B.f. $f$ on $n+1$ variables can be expressed in the form  
\begin{align}\label{factoring} f=x_{n+1}g(x_1,...,x_n)+h(x_1,...,x_n).\end{align}
We show that if the weights of the functions $g$ and $h$ on $n$ variables are known, then the weight of a B.f. on $n+1$ variables is obtained.

\begin{theorem}\label{weight-factoring}
	Let $f=x_{n+1}g(x_1,...,x_n)+h(x_1,...,x_n)$, with $f\in B_{n+1}$ and \\$g,h\in B_n$. Then \\(i) \(\mathrm{w}(f)=\mathrm{w}((g+h)_{\restriction\mathbb{F}^n})+\mathrm{w}(h_{\restriction\mathbb{F}^n}),\) \\(ii) $f$ is balanced if both $g+h$ and $h$ are balanced, \\(iii) $f$ is unbalanced if one in $\{g+h,h\}$ is balanced and the other is not.	
\end{theorem}

\begin{proof} In this proof we view $h$, $g$ and $g+h$ as functions in $B_n$.
	
	(i) Let $X=(x_1,...,x_n)$. We have 
	\begin{align}\label{fourier-factoring-eqn}
	\mathcal{F}(f)&=\sum_{(X,x_{n+1})\in\mathbb{F}^n\times \mathbb{F}}(-1)^{x_{n+1}g(X)+h(X)}\nonumber\\&=\sum_{X\in\mathbb{F}^n}(-1)^{g(X)+h(X)}+\sum_{X\in\mathbb{F}^n}(-1)^{h(X)}\nonumber\\&=\mathcal{F}(g+h)+\mathcal{F}(h)
	\end{align}
	
	So
	\begin{align}\label{weight-factoring-eqn}
		{\rm w}(f)&=2^n-\frac{1}{2}\mathcal{F}(f)\nonumber\\&=2^n-\frac{1}{2}\left[\mathcal{F}(g+h)+\mathcal{F}(h)\right]\\&=2^n-\frac{1}{2}\left[2^n-2{\rm w}(g+h)+2^n-2{\rm w}(h)\right]\nonumber\\&=\mathrm{w}(g+h)+\mathrm{w}(h).\nonumber
	\end{align}
	
	(ii) Observe that if $g+h$ and $h$ are balanced, we have $\mathcal{F}(g+h)=\mathcal{F}(h)=0$ which implies that $\mathrm{w}(f)=2^n$, by Equation \eqref{weight-factoring-eqn}.
	
	(iii) Without loss of generality, suppose that $g+h$ is balanced while $h$ not. Then $\mathcal{F}(g+h)=0$ and $\mathcal{F}(h)\neq 0$. So, by Equation \eqref{weight-factoring-eqn}, we have $\mathrm{w}(f)=2^n-\frac{1}{2}\mathcal{F}(h)\neq 2^n$ since $\mathcal{F}(h)\neq 0$, and so $f$ is unbalanced. 
\end{proof}

Our first two constructions of balanced B.f.'s are based on the well-known fact in the Proposition~\ref{balanced-splitting} and Theorem~\ref{weight-factoring}.

\begin{proposition}\label{balanced-construction-1}
	Let \(f\sim_A x_{n+1}g(x_1,...,x_n)+h(x_1,...,x_{n-1})\), where \\$g=\tilde{g}(x_1,...,x_{n-1})+x_n$ and $h=\tilde{h}(x_1,...,x_{n-2})+x_{n-1}$. Then $f$ is balanced.	
\end{proposition}

\begin{proof}
	By Proposition \ref{balanced-splitting}, both $g+h$ and $h$ are balanced and so, applying Theorem~\ref{weight-factoring}, $f$ is balanced.
\end{proof}

Notice that the result which we present in the following proposition is partly an extension of Proposition~\ref{balanced-construction-1}.

\begin{proposition}\label{balanced-construction-2}
	Let $g_i=\tilde{g}_i(x_{i+1},...,x_{n-i})+x_{n-i+1}$ be a B.f. on $n-2i+1$ variables, with integer $n>2$ and $1\leq i\leq \lfloor\frac{n}{2}\rfloor$, and define the two functions on $n$ variables as:
	\begin{align}\label{construction-2-eqn-1}
	f_\ell\sim_A\sum_{i=1}^{\ell-1}x_ig_i+g_\ell
	\end{align} and 
	\begin{align} 
	\bar{f}_\ell\sim_A\sum_{i=1}^{\ell}x_ig_i+c,
	\end{align}\label{construction-2-eqn-2}
	with $\ell\leq \lfloor\frac{n}{2}\rfloor$ and $c\in\mathbb{F}$. Then $f_\ell$ is balanced and $\bar{f}_\ell$ is unbalanced. 
\end{proposition}

\begin{proof}
	For a positive integer $t\leq \ell-1$, define \[h_t=\sum_{i=t}^{\ell-1}x_ig_i+g_\ell \hspace{1cm}\text{ and } \hspace{1cm} \bar{h}_t=\sum_{i=t}^{\ell}x_ig_i+c,\] with $c\in\mathbb{F}$. Since $\mathcal{F}(f_\ell)$ is invariant under affine equivalence (see Remark~\ref{fourier-affine-equivalence-invariant}) then, by Equation~\ref{fourier-factoring-eqn}, we obtain 
	
	\begin{align}\label{bal-fourier-eqn}
	\mathcal{F}(f_\ell)=\sum_{i=1}^{\ell-2}\mathcal{F}(g_i+h_{i+1})+\mathcal{F}(g_{\ell-1}+g_{\ell})+\mathcal{F}(g_\ell)
	\end{align}
	and 
	\begin{align}\label{bal-fourier-eqn-1}
	\mathcal{F}(\bar{f}_\ell)=\sum_{i=1}^{\ell-1}\mathcal{F}(g_i+\bar{h}_{i+1})+\mathcal{F}(g_\ell+c)+\mathcal{F}(c).
	\end{align}
	We conclude by Proposition~\ref{balanced-splitting} that  $g_i+h_{i+1}$, $g_i+\bar{h}_{i+1}$, $g_{\ell-1}+g_{\ell}$ and $g_\ell+c$ are all balanced. So it implies that
	\[\mathcal{F}(g_\ell+c)=\mathcal{F}(g_{\ell-1}+g_\ell)=\mathcal{F}(g_i+h_{i+1})=\mathcal{F}(g_i+\bar{h}_{i+1})=0.\] 
	It follows that Equation~\eqref{bal-fourier-eqn} becomes \(\mathcal{F}(f_\ell)=0\), implying that $f_\ell$ is balanced and Equation \eqref{bal-fourier-eqn-1} becomes $\mathcal{F}(\bar{f}_\ell)=\mathcal{F}(c)\neq 0$ which implies that $\bar{f}_\ell$ is unbalanced.	
\end{proof}

\begin{remark}
	All the quadratic B.f.'s are a special case of the functions constructed in Proposition~\ref{balanced-construction-2} since if we let $\tilde{g}_i=0$, for all $1\leq i\leq \ell$, we obtain their classification via affine equivalence as given in Theorem \ref{quadratic}.
\end{remark}

Any B.f. can also be expressed in the form \begin{align}\label{convolution-product}f=x_{n+1}g(x_1,...,x_n)+(1+x_{n+1})h(x_1,...,x_n).\end{align} We call this form the {\em convolutional product} of $g$ and $h$. 

Next we completely classify the balanced cubic functions of the class  $f=x_{n+1}g(x_1,...,x_n)+(1+x_{n+1})h(x_1,...,x_n)$, with $\deg(h),\deg(g)\leq 2$.

\begin{theorem}\label{balanced-cubic}
	Let $f=x_{n+1}g(x_1,...,x_n)+(1+x_{n+1})h(x_1,...,x_n)$ on $n+1$ variables, with $\deg(h),\deg(g)\leq 2$, be cubic. Then $f$ is balanced if and only if either both $g$ and $h$ are balanced or $g=h\circ \varphi +1$, for some affinity $\varphi$ and with both $g$ and $h$ unbalanced quadratic.
\end{theorem}

\begin{proof}
	Recall that $\mathcal{F}(g)=2^n-2\mathrm{w}(g)$ and by Equation \eqref{fourier-factoring-eqn}, we have \(\mathcal{F}(f)=\mathcal{F}(g_{\restriction\mathbb{F}^n})+\mathcal{F}(h_{\restriction\mathbb{F}^n})\). So $f$ is balanced $\iff \mathcal{F}(f)=0\iff \mathcal{F}(g_{\restriction\mathbb{F}^n})=-\mathcal{F}(h_{\restriction\mathbb{F}^n})\iff 2^n-2\mathrm{w}(g_{\restriction\mathbb{F}^n})=-2^n+2\mathrm{w}(h_{\restriction\mathbb{F}^n})\iff \mathrm{w}(g_{\restriction\mathbb{F}^n})+\mathrm{w}(h_{\restriction\mathbb{F}^n})=2^n\iff \mathrm{w}(g_{\restriction\mathbb{F}^n})=2^n-\mathrm{w}(h_{\restriction\mathbb{F}^n})\iff \mathrm{w}(g_{\restriction\mathbb{F}^n})=\mathrm{w}(h_{\restriction\mathbb{F}^n}+1)\iff$  either both $g$ and $h$ are balanced or both  $g$ and $h$ unbalanced quadratics related by $g=h\circ \varphi +1$, for some affinity $\varphi$ (see Lemma \ref{weight-affine-equivalent-quadratics}).
\end{proof}

Observe that the forward direction of Theorem \ref{balanced-cubic} holds in general but its converse might not be necessarily always true.

In the next result we construct balanced functions based on bent functions.
\begin{proposition}\label{balanced-bent}
	Let $f=x_{n+1}g(x_1,...,x_n)+(1+x_{n+1})h(x_1,...,x_n)$, with $n$ even, be a B.f. on $\mathbb{F}^{n+1}$ such that $g$ and $h$ are both bent. Then $f$ is balanced if and only if $\mathrm{w}(g)\neq \mathrm{w}(h)$.
\end{proposition}

\begin{proof}
	Since $\mathcal{F}(g)=W_g(0)=\pm 2^{\frac{n}{2}}$, so any bent function on $\mathbb{F}^n$ has the weight $2^{n-1}\pm 2^{\frac{n}{2}-1}$. Since \(\mathrm{w}(f)=\mathrm{w}(g_{\restriction\mathbb{F}^n})+\mathrm{w}(h_{\restriction\mathbb{F}^n})\), so $\mathrm{w}(f)=2^n\pm 2^{\frac{n}{2}}$ if  $\mathrm{w}(g_{\restriction\mathbb{F}^n})=\mathrm{w}(h_{\restriction\mathbb{F}^n})$ and $\mathrm{w}(f)=2^n$ if $\mathrm{w}(g_{\restriction\mathbb{F}^n})\neq \mathrm{w}(h_{\restriction\mathbb{F}^n})$. Hence $f$ is balanced if and only if $\mathrm{w}(g)\neq \mathrm{w}(h)$. 
\end{proof}

Next we show that the balanced function in Proposition \ref{balanced-bent} [also for the unbalanced, that is, if $\mathrm{w}(g)=\mathrm{w}(h)$] are in fact plateaued.

\begin{proposition}
	Let $f=x_{n+1}g(x_1,...,x_n)+(1+x_{n+1})h(x_1,...,x_n)$, with $n$ even, be a B.f. on $\mathbb{F}^{n+1}$ such that $g$ and $h$ are both bent. Then $f$ is a plateaued function.
\end{proposition}

\begin{proof}
	Let $\alpha=(a,a_{n+1})\in\mathbb{F}^n\times\mathbb{F}$ and $z=(X,x_{n+1})\in\mathbb{F}^n\times\mathbb{F}$, where $X=(x_1,...,x_n)$. Then we have
	\begin{align}\label{walsh}
	W_f(\alpha)&=\sum_{z\in\mathbb{F}^{n+1}}(-1)^{f(z)+\alpha\cdot z}\nonumber\\& =\sum_{(x_{n+1},X)\in\mathbb{F}\times\mathbb{F}^n}(-1)^{x_{n+1}g(X)+(1+x_{n+1})h(X)+a\cdot X+a_{n+1}\cdot x_{n+1}}\nonumber\\&=\sum_{X\in\mathbb{F}^n}(-1)^{h(X)+a\cdot x}+\sum_{X\in\mathbb{F}^n}(-1)^{g(X)+a\cdot X+a_{n+1}}\nonumber\\&= W_{h_{\restriction\mathbb{F}^n}}(a)+(-1)^{a_{n+1}}W_{g_{\restriction\mathbb{F}^n}}(a).
	\end{align}
	
	Since $g$ and $h$ are bent then, for any $a\in\mathbb{F}^n$, the only possible values for $W_{h_{\restriction\mathbb{F}^n}}(a)$ and $W_{g_{\restriction\mathbb{F}^n}}(a)$ are $\pm 2^{\frac{n}{2}}$. So, for any $\alpha=(a,a_{n+1})\in\mathbb{F}^n\times\mathbb{F}$, $W_f(\alpha)$ takes one of the values $0$ or $\pm 2^{\frac{n}{2}+1}$. Hence $f$ is plateaued.
\end{proof}

\section{Linear space of balanced Boolean functions}\label{linear structures-section}
We present some conditions which help to determine whether a derivative of a B.f. is constant and we utilise them to check the balanced B.f.'s, constructed in Section~\ref{balanced-section}, whose linear space is trivial. 

\begin{proposition}\label{derivative-first}
	Let $f=x_{n+1}g(x_1,...,x_n)+h(x_1,...,x_n)$, where $g,h\in B_n$ and $f\in B_{n+1}$. Let 
	$\lambda=(a,a_{n+1})\in\mathbb{F}^n\times\mathbb{F}$. 
	Then \begin{align*}D_\lambda f\sim_A x_{n+1}D_ag+a_{n+1}g+D_ah.\end{align*}
\end{proposition}

\begin{proof}
	Let $X=(x_1,...,x_n)\in \mathbb{F}^n$. Thus, we have $f=x_{n+1}g(X)+h(X)$. Let $\lambda=(a,a_{n+1})\in\mathbb{F}^n\times\mathbb{F}$. So \begin{align*}
	D_\lambda f&=(x_{n+1}+a_{n+1})g(X+a)+h(X+a)+x_{n+1}g(X)+h(X)\\&=x_{n+1}\left[g(X+a)+g(X)\right]+a_{n+1}g(X+a)+h(X+a)+h(X)\\&=x_{n+1}D_ag(X)+a_{n+1}[D_ag(X)+g(X)]+D_ah(X)\\&\sim_A x_{n+1}D_ag(X)+a_{n+1}g(X)+D_ah(X). \hspace{1cm}(\text{apply } x_{n+1}\mapsto x_{n+1}+a_{n+1})\qedhere
	\end{align*}
\end{proof}

	For $f\in B_n$, we define the set which contains all $a\in\mathbb{F}^n$ such that $D_af$ is balanced by $\Gamma(f)$, that is, $\Gamma(f)=\{a\in\mathbb{F}^n\mid D_af \text{ is balanced}\}$ (see \cite{Cal}). 
	
	We next show that the linear space of B.f. $f$ and $\Gamma(f)$ are both invariant under affine equivalence.

\begin{lemma}\label{equivalence-derivatives}
	Let $g_1,g_2\in B_n$ be such that $g_1\sim_A g_2$. Then $|V(g_1)|=|V(g_2)|$ and $|\Gamma(g_1)|=|\Gamma(g_2)|$.
\end{lemma}

\begin{proof}
	Let $\varphi$ be the affinity of $\mathbb{F}^n$ associated with invertible $M\in GL_n(\mathbb{F})$ and $w\in\mathbb{F}^n$, that is, $\varphi(y)=M\cdot y+w$, for all $y\in \mathbb{F}^n$. For  $a\in \mathbb{F}^n$, we have 
	\begin{align}\label{equivalence-derivatives-eqn}
	D_ag_1(x)&=D_a(g_2\circ \varphi)(x)\nonumber\\&=g_2(\varphi(x+a))+g_2(\varphi(x))\nonumber\\&=g_2(M\cdot (x+a)+w)+g_2(\varphi(x))\nonumber\\&=g_2(M\cdot x+M\cdot a+w)+g_2(\varphi(x))\nonumber\\&=g_2(M\cdot a+\varphi(x))+g_2(\varphi(x))\nonumber\\&=D_{M\cdot a}g_2(\varphi(x))=(D_{M\cdot a}g_2\circ \varphi)(x).
	\end{align}
	So it implies that $D_ag_1=(D_{M\cdot a}g_2)\circ \varphi\sim_A D_{M\cdot a}g_2$. It follows by Proposition \ref{equivalence-properties} that ${\rm w}(D_ag_1)={\rm w}(D_{M\cdot a}g_2)$, so we conclude that $D_ag_1$ is balanced if and only if $D_{M\cdot a}g_2$ is balanced, $D_ag_1=0$ if and only if $D_{M\cdot a}g_2\sim_A 0$,  and $D_ag_1=1$ if and only if $D_{M\cdot a}g_2\sim_A 1$. Hence we have $|V(g_1)|=|V(g_2)|$ and $|\Gamma(g_1)|=|\Gamma(g_2)|$.
\end{proof}

\begin{proposition}\label{derivative-first-constant}
	Let $f=x_{n+1}g(x_1,...,x_n)+h(x_1,...,x_n)$, where $g,h\in B_n$ and $f\in B_{n+1}$. Let 
	$\lambda=(a,a_{n+1})\in\mathbb{F}^n\times\mathbb{F}$. Then $D_\lambda f=c$, with $c\in\mathbb{F}$ (i.e., $D_\lambda f$ is constant) if and only if $D_ag=0$ and $D_ah=a_{n+1}g+c$.		
\end{proposition}

\begin{proof}
	 $D_\lambda f=c$, with $c\in\mathbb{F}$ (i.e., $D_\lambda f$ is constant) if and only if \[x_{n+1}D_ag+a_{n+1}g+D_ah=c\] (see Proposition \ref{derivative-first}) if and only if $D_ag=0$ and $D_ah=a_{n+1}g+c$.
\end{proof}

We can deduce from Proposition \ref{derivative-first-constant} that the following result holds.

\begin{corollary}\label{derivative-first-nonconstant}
	Let $f=x_{n+1}g(x_1,...,x_n)+h(x_1,...,x_n)$, where $g,h\in B_n$ are non-constant and $f\in B_{n+1}$. Let 
	$\lambda=(a,a_{n+1})\in\mathbb{F}^n\times\mathbb{F}$. Then $D_\lambda f$ is non-constant if and only if one of the following happens: 
	\begin{itemize}
		\item[(i)] $D_ag\neq 0$,
		\item[(ii)] $D_ag=0$ and $D_ah\neq a_{n+1}g+c$, with $c\in\mathbb{F}$. 		
	\end{itemize}
\end{corollary}

\begin{proposition}
	If $f=x_{n+1}g(x_1,...,x_n)+h(x_1,...,x_n)$, with $n$ even and $g$ bent, then $f$ has a trivial linear space.
\end{proposition}

\begin{proof}
	Suppose that $g$ is a bent function and let $\lambda=(a,a_{n+1})\in\mathbb{F}^n\times\mathbb{F}$.	By Proposition \ref{derivative-first}, we have  \(D_\lambda f\sim_A x_{n+1}D_ag+a_{n+1}g+D_ah.\) Observe that when $\lambda=(0,1)$ we have $D_\lambda f\sim_A g$ which is a non-constant function since $g$ is bent. If we show that $D_\lambda f$ is non-constant, for all $\lambda=(a,a_{n+1})\in(\mathbb{F}^n\times\{0\})\times\mathbb{F}$, then we are done. Since $g$ is bent then $D_ag$ is balanced (i.e. nonzero), for any $a\in\mathbb{F}^n\setminus\{0\}$, and so we conclude by Corollary \ref{derivative-first-nonconstant}(i) that $D_\lambda f$ is non-constant, for all $\lambda=(a,a_{n+1})\in(\mathbb{F}^n\times\{0\})\times\mathbb{F}$.
\end{proof}

In the next result, we apply Corollary \ref{derivative-first-nonconstant} to show that some balanced functions constructed in Proposition \ref{balanced-construction-1} have trivial linear space. 

\begin{proposition}\label{balanced-nonconstant-derivatives}
	Let $f$ be as constructed in Proposition \ref{balanced-construction-1}. If $n\geq 3$ is odd and $\tilde{g}$ with restriction to $\mathbb{F}^{n-1}$ is bent, then the linear space of $f$ is trivial.
\end{proposition}  

\begin{proof}
	Assume that $\tilde{g}$, with restriction to $\mathbb{F}^{n-1}$, is bent and let $\lambda=(a,a_{n+1})\in\mathbb{F}^n\times\mathbb{F}$.	We know, by Proposition \ref{derivative-first}, that  \(D_\lambda f\sim_A x_{n+1}D_ag+a_{n+1}g+D_ah.\) Observe that when $\lambda=(0,1)$ we have $D_\lambda f\sim_A g$ which is clearly non-constant as $\tilde{g}$ is bent. Now we remain to show that $D_\lambda f$ is non-constant, for all $\lambda=(a,a_{n+1})\in(\mathbb{F}^n\setminus\{0\})\times\mathbb{F}$. We know from Corollary \ref{derivative-first-nonconstant} that if $D_ag$ is nonzero then $D_\lambda f$ is non-constant. So we can simply show that $D_a g$ is nonzero, for all $a\in\mathbb{F}^n\setminus\{0\}$.
	
	Let $a=(\tilde{a},a_n)\in\mathbb{F}^{n-1}\times\mathbb{F}$, where $\tilde{a}=(a_1,...,a_{n-1})$. If $\tilde{a}=(0,...,0)$ and $a_n=1$, then we have $D_a g=1$ which is nonzero. If $a=(\tilde{a},1)$, with $\tilde{a}\in\mathbb{F}^{n-1}\setminus\{0\}$, we have $D_ag=D_{\tilde{a}}\tilde{g}+1$ which must be nonzero as $D_{\tilde{a}}\tilde{g}$ is balanced because $\tilde{g}$ is bent. If $a=(\tilde{a},0)$, with $\tilde{a}\in\mathbb{F}^{n-1}\setminus\{0\}$, we have $D_ag=D_{\tilde{a}}\tilde{g}$ which is balanced as $\tilde{g}$ is bent. Thus, $D_ag$ is nonzero, for all $a\in\mathbb{F}^n\setminus\{0\}$. Hence the linear space of $f$ is trivial. 
\end{proof}

Notice that we can apply similar arguments as in the proof of Proposition~\ref{balanced-nonconstant-derivatives} to show that the linear space for any function of the form given in Proposition~\ref{balanced-construction-2}, with $\tilde{g}_1$ bent, is trivial.

\begin{example}
	For any positive odd integer $n\geq 3$, a function of the form: \[f=x_{n+1}(x_1x_2+\cdots +x_{n-2}x_{n-1}+x_n)+h(x_1,...,x_{n-2})+x_{n-1}\] is balanced and its linear space is trivial.
\end{example}

Next we determine whether the linear space of any balanced cubic function of the form \eqref{convolution-product} [i.e., $f=x_{n+1}g(x_1,...,x_n)+(1+x_{n+1})h(x_1,...,x_n)$, with $\deg(g),\deg(h)\leq 2$] is trivial. From Theorem \ref{balanced-cubic}, we know that such functions are balanced if and only if either both $g$ and $h$ are balanced or $g=h\circ \varphi +1$, for some unbalanced quadratics $g$ and $h$, and an affinity $\varphi$.

\begin{proposition}\label{balanced-quadratic-bent-derivatives}
	Let $f=x_{n+1}g(x_1,...,x_n)+(1+x_{n+1})h(x_1,...,x_n)$ on $\mathbb{F}^{n+1}$, with $n$ even, be cubic such that $g$ and $h$ are quadratic bent related by $g=h\circ \varphi +1$, for some affinity $\varphi$. Then the linear space of $f$ is trivial.
\end{proposition}

\begin{proof}
	Suppose that both $g$ and $h$, with restrictions to $\mathbb{F}^n$, are bent. Let $\lambda=(a,a_{n+1})\in\mathbb{F}^n\times\mathbb{F}$. Observe that \(f=x_{n+1}(g+h)+h\), and so $f$ is cubic if and only if $g+h$ is a quadratic function. So we assume that $g+h$ is quadratic. By Proposition \ref{derivative-first}, we have  \(D_\lambda f\sim_A x_{n+1}D_a(g+h)+a_{n+1}(g+h)+D_ah.\) Observe that when $\lambda=(1,0)$ we have $D_\lambda f\sim_A g+h$ which is non-constant as we assumed that $g+h$ is quadratic. 
	
	Next we prove that $D_\lambda f$ is non-constant, for all $\lambda=(a,a_{n+1})\in(\mathbb{F}^n\setminus\{0\})\times\mathbb{F}$. By Corollary \ref{derivative-first-nonconstant}(i), we know that if $D_a(g+h)\neq 0$, then $D_\lambda f$ is non-constant. Now we show that $D_\lambda f$ is still non-constant if $D_a(g+h)=0$, for some $a\in\mathbb{F}^n\setminus\{0\}$.  Assume that $D_a(g+h)=0$, for some $a\in\mathbb{F}^n\setminus\{0\}$. Then we have \(D_\lambda f\sim_A a_{n+1}(g+h)+D_ah.\) If $a_{n+1}=0$ then \(D_\lambda f\sim_A D_ah\), and so it is non-constant since $D_ah$ has to be balanced as $h$ is bent. If $a_{n+1}=1$ then \(D_\lambda f\sim_A g+h+D_ah\) which is also non-constant since $g+h$ is a quadratic and $D_ah$ has degree $1$ as it is balanced.
\end{proof}

\begin{remark}\label{degree-conv-prod}
	Since the convolutional product of $g$ and $h$ can be reduced to $f=x_{n+1}(g+h)+h$, so either $\deg(f)=\deg(h)$ [this happens when $\deg(g+h)<\deg(h)$] or $\deg(f)=\max\{\deg(g),\deg(h)\}+1$. Moreover, we can use Theorem~\ref{weight-factoring}, to deduce that $\mathrm{w}(f)=\mathrm{w}(g_{\restriction\mathbb{F}^n})+\mathrm{w}(h_{\restriction\mathbb{F}^n})$ and $f$ is balanced if $g$ and $h$ are balanced.
\end{remark}

Finally, we determine some balanced functions constructed in Proposition~\ref{balanced-bent} [i.e., $f=x_{n+1}g(x_1,...,x_n)+(1+x_{n+1})h(x_1,...,x_n)$, where $g$ and $h$ are both bent and $\mathrm{w}(g)\neq \mathrm{w}(h)$] which have trivial linear space.

\begin{proposition}\label{balanced-bent-derivatives}
	Let $f=x_{n+1}g(x_1,...,x_n)+(1+x_{n+1})h(x_1,...,x_n)$, with $n$ even, be a B.f. on $n+1$ variables such that $g$ and $h$ are both bent. Then the linear space of $f$ is trivial if $\deg(f)=\max\{\deg(g),\deg(h)\}+1$.
\end{proposition}

\begin{proof}
	Recall that \(D_\lambda f\sim_A x_{n+1}D_a(g+h)+a_{n+1}(g+h)+D_ah\), for $\lambda=(a,a_{n+1})\in\mathbb{F}^n\times\mathbb{F}$ (see Proposition~\ref{balanced-quadratic-bent-derivatives}). Observe that $f=x_{n+1}(g+h)+h$. We are given that $\deg(f)=\max\{\deg(g),\deg(h)\}+1$. So it follows that $\deg(g+h)=\max\{\deg(g),\deg(h)\}$, implying that $g+h$ is non-constant since $g$ and $h$ are bent. When $\lambda=(0,1)$, we have $D_\lambda f\sim_A g+h$ which is non-constant. 
	
	Now we prove that $D_\lambda f$, for all $\lambda=(a,a_{n+1})\in(\mathbb{F}^n\setminus\{0\})\times\mathbb{F}$, is non-constant. If $D_a(g+h)\neq 0$, then $D_\lambda f$ is non-constant, by Corollary \ref{derivative-first-nonconstant}(i). Suppose that $D_b(g+h)=0$, for some $b\in\mathbb{F}^n\setminus\{0\}$. We need to show that $D_\lambda f$ is still non-constant, for $\lambda=(b,a_{n+1})\in(\mathbb{F}^n\setminus\{0\})\times\mathbb{F}$. In this case we have \(D_\lambda f\sim_A a_{n+1}(g+h)+D_bh.\) If $a_{n+1}=0$ then we have \(D_\lambda f\sim_A D_bh\) which is non-constant since $D_bh$ has to be balanced as $h$ is bent. If $a_{n+1}=1$ then we have \(D_\lambda f\sim_A g+h+D_bh\). Since $\deg(g+h)=\max\{\deg(g),\deg(h)\}$, so we have $\deg(g+h)=\max\{\deg(g),\deg(h)\}>\deg(D_bh)$, implying that $\deg(D_\lambda f)=\deg(g+h)>\deg(D_bh)$. So $D_\lambda f$ must be non-constant. Hence the linear space of $f$ is trivial. 
\end{proof}

Let $\alpha=(a,a_{n+1})\in\mathbb{F}^n\times\mathbb{F}$. Observe that, from Equation \eqref{walsh}, we obtain $|W_f(\alpha)|\leq |W_{g_{\restriction\mathbb{F}^n}}|+|W_{h_{\restriction\mathbb{F}^n}}|$. So it follow that the nonlinearity of $f$ in Propositions~ \ref{balanced-quadratic-bent-derivatives} and \ref{balanced-bent-derivatives} is
\begin{align}\label{nonlinearity-eqn}
\mathcal{N}(f)&=2^n-\frac{1}{2}\max_{\alpha\in\mathbb{F}^{n+1}}|W_f(\alpha)|\nonumber\\&\geq 2^n-\frac{1}{2}\max_{\alpha\in\mathbb{F}^{n+1}}(|W_{g_{\restriction\mathbb{F}^n}}|+|W_{h_{\restriction\mathbb{F}^n}}|)\nonumber\\&\geq 2^{n-1}-\frac{1}{2}\max_{\alpha\in\mathbb{F}^{n+1}}|W_{g_{\restriction\mathbb{F}^n}}|+2^{n-1}-\frac{1}{2}\max_{\alpha\in\mathbb{F}^{n+1}}|W_{h_{\restriction\mathbb{F}^n}}|\nonumber\\&=\mathcal{N}(g_{\restriction\mathbb{F}^n})+\mathcal{N}(h_{\restriction\mathbb{F}^n})
\end{align}
This suggests a way of constructing B.f.'s with high nonlinearity. For instance, from the relation \eqref{nonlinearity-eqn}, we deduce that the nonlinearity of the balanced function $f$ constructed in Propositions \ref{balanced-quadratic-bent-derivatives} and \ref{balanced-bent-derivatives} is $\mathcal{N}(f)\geq 2^{N-1}-2^{\frac{N-1}{2}}$, with $N=n+1$.

\begin{example}
	Let $g=x_1x_2+x_3x_4+1$ and $h=x_1x_4+x_2x_3$. The cubic function $f=x_5g+(1+x_5)h$ is balanced and its linear space is trivial. It can be easily verified that $\mathcal{N}(f)=12$, implying that $f$ is semi-bent. Note that both $g$ and $h$, with restriction to $\mathbb{F}^4$, are bent related by $g=h\circ \varphi +1$, where \(\varphi=A(x_1,x_2,x_3,x_4)^T\),  with \[A=\left(\begin{array}{cccc}
	1 & 0 & 0 & 0 \\ 
	0 & 0 & 0 & 1 \\ 
	0 & 0 & 1 & 0 \\ 
	0 & 1 & 0 & 0
	\end{array}\right).\]
\end{example}

\section{APN functions in even dimension}\label{APN functions-section}
In this section we study the linear spaces of components of APN functions in even dimension. We show that, for any APN function, there must be a component with trivial linear space. We also provide a general form for the number of bent components in quadratic APN function and show bounds on their number. 
\subsection{Linear space for components of APN functions in even dimension} 
We first give some definitions and results which are crucial in studying the linear spaces of components of APN functions in even dimension.   
 
\begin{definition}\label{splitting-function}
	A B.f. $f$ on $n$ variables is called a {\em splitting function} if we have $f\sim_A g(x_1,...,x_i)+h(x_{i+1},...,x_n)$, for some positive integer $i$, $g\in B_i$ and $h\in B_{n-i}$. We say that $i$ is a {\em splitting number} of $f$ and $S(f)$ denotes the set of all splitting numbers of $f$. We define a {\em splitting index} of $f$ as the number $\sigma(f)=\min S(f)$.
\end{definition}

\begin{remark}
	Let $f\in B_n$ be a splitting function. Then 
	\begin{itemize}
		\item[1.] $i$ is a splitting number $\iff$ $n-i$ is a splitting number,
		\item[2.] clearly, $\sigma(f)\in\{1,...,\lfloor n/2\rfloor\}$.
	\end{itemize}
\end{remark}

\begin{lemma}\label{linear-subspace-splitting}
	Let $f\in B_n$. Then $\sigma(f)=1$ if and only if $\dim V(f)\geq 1$. 
\end{lemma}

\begin{proof}
	Suppose that $\sigma(f)=1$, that is, $f\sim_A \tilde{f}=g(x_1)+h(x_2,...,x_n)$. So we have $\tilde{f}=f\circ \varphi$, where $\varphi(y)=My+w$, for some $w\in\mathbb{F}^n$ and invertible $M\in GL_n(\mathbb{F})$. Clearly $D_{e_1}\tilde{f}$ is constant. By Equation \eqref{equivalence-derivatives-eqn} in the proof Lemma \ref{equivalence-derivatives}, we have $D_{e_1}\tilde{f}=(D_{Me_1}f)\circ \varphi$ and since $\mathrm{w}(D_{e_1}\tilde{f})=\mathrm{w}(D_{Me_1}f)$ (see Proposition~\ref{equivalence-properties}), so $(D_{Me_1}f)\circ \varphi$ must also be constant. Note that $Me_1\neq 0$ since $M$ is a linear isomorphism. Thus both $0$ and $Me_1$ are in $V(f)$ which implies that $\dim V(f)\geq 1$.
	
	Conversely, suppose that $\dim V(f)\geq 1$, that is, $\exists a\neq 0\in V(f)$ such that $D_af=c$, with $c\in \mathbb{F}$. We can take the $\mathbb{F}$-linear isomorphism $E$ of $\mathbb{F}^n$ that sends $e_1\mapsto Ee_1=a$ so that we have $\tilde{f}=f\circ E$ and thus, \[D_{e_1}\tilde{f}=(D_{Ee_1}f)\circ E=(D_af)\circ E=(c)\circ E=c\] which implies that $D_{e_1}\tilde{f}$ is constant. Since we have $D_{e_1}\tilde{f}=c$, so we can write $\tilde{f}=cx_1+h(x_2,...,x_n)$. Hence $\sigma(f)=1$ as $f\sim_A \tilde{f}$.
\end{proof}

\begin{remark}\label{split-weight-fourier}
	If $g(x_1,...,x_s)$, with a positive integer $s<n$, is in $B_n$ then we have $\mathrm{w}(g)=2^{n-s}\mathrm{w}(g_{\restriction\mathbb{F}^s})$.
\end{remark}

The preceding remark is useful in the following.

\begin{lemma}\label{balanced-split-1}
		Let $f\in B_n$, with $n$ even. If $\sigma(f)=1$, then \(|\Gamma(f)|\leq 2^n-4.\)
\end{lemma}

\begin{proof}
	Suppose $f\in\mathbb{F}^n$ has $\sigma(f)=1$, that is, $f\sim_A \tilde{f}=cx_1+h(x_2,...,x_n)$, with $c\in\mathbb{F}$. By Lemma \ref{equivalence-derivatives}, we have $|\Gamma(f)|=|\Gamma(\tilde{f})|$, so we can simply consider $|\Gamma(\tilde{f})|$. It is clear that $0$ and $e_1$ are both not in $\Gamma(\tilde{f})$ since $D_0\tilde{f}=0$ and $D_{e_1}\tilde{f}=c$. Suppose that these are the only ones, that is, $|\Gamma(\tilde{f})|= 2^n-2$. This implies that, for all $a\in\mathbb{F}^n\setminus\{0,e_1\}$, $D_a\tilde{f}$ is balanced. 
	
	Let $W=<e_2,...,e_n>$ and denote $W^*=W\setminus\{0\}$.  Clearly, $W^*$ is contained in $\mathbb{F}^n\setminus\{0,e_1\}$, that is, $W^*\subset \Gamma(\tilde{f})$. So, for all $a\in W^*$, $D_a\tilde{f}$ is balanced. It is clear that $W\simeq \mathbb{F}^{n-1}$. Observe that, for any $a=(0,b)\in\{0\}\times(\mathbb{F}^{n-1}\setminus\{0\})=W^*$, we have $D_a\tilde{f}=D_bh$ as the first coordinate of $a$ is $0$. Since $D_a\tilde{f}$ does not depend on $x_1$ then, by Remark~\ref{split-weight-fourier}, we have $2^{n-1}=\mathrm{w}(D_a\tilde{f})=\mathrm{w}(D_bh)=2\mathrm{w}(D_bh_{\restriction\mathbb{F}^{n-1}}) \implies \mathrm{w}(D_bh_{\restriction\mathbb{F}^{n-1}})=2^{n-2}$, that is, $D_bh_{\restriction \mathbb{F}^{n-1}}$ is balanced, for all $b\in \mathbb{F}^{n-1}\setminus \{0\}$. This implies that $h$, with restriction to $\mathbb{F}^{n-1}$, is bent (see Theorem \ref{bent-thm}).
	
	But $n-1$ is odd as $n$ is even, so it implies that we have a bent function on $\mathbb{F}$-vector space of odd dimension, which is impossible. Thus, the assumption that $|\Gamma(\tilde{f})|=2^n-2$ is false, and so we can say $|\Gamma(\tilde{f})|\leq 2^n-3$.
	
	Suppose that $d\in\mathbb{F}^n\setminus\{0,e_1\}$ is the other nonzero element such that $D_d\tilde{f}$ is unbalanced. So $D_{d+e_1}\tilde{f}(x)=D_{e_1}\tilde{f}(x)+D_df(x+e_1)=c+D_df(x+e_1)=(c+D_df(x))\circ \varphi$, for $c\in\mathbb{F}$ and $\varphi(y)=Iy+e_1$, with $I$ as an identity in $GL_n(\mathbb{F})$. That is, $D_{d+e_1}\tilde{f}(x)\sim_A D_df(x)+c$. Since $D_df(x)$ is unbalanced then $D_df(x)+c$ must be unbalanced, implying that $D_{d+e_1}\tilde{f}(x)$ is also unbalanced. That is, $\{0,e_1, d, d+e_1\}\not\subset \Gamma(\tilde{f})$. Hence we have $|\Gamma(\tilde{f})|\leq 2^n-4$.
\end{proof}

Next we state a well-known result for characterization of APN function.
 
\begin{theorem}[\cite{Ber}]\label{APN}
	Let $F$ be a v.B.f. from $\mathbb{F}^n$ into $\mathbb{F}^n$. Then
	\begin{align}
	\sum_{\lambda\neq 0\in\mathbb{F}^n}\sum_{a\in\mathbb{F}^n}\mathcal{F}^2(D_a(F_\lambda))\geq 2^{2n+1}(2^n-1).
	\end{align}
	Moreover, $F$ is APN if and only if equality holds.
\end{theorem} 

In the next results we discuss about the linear space for components of an APN function in even dimension. 

\begin{theorem}\label{APN-components-linear-space}
	Let a v.B.f. $F$ from $\mathbb{F}^n$ to itself, with $n$ even, be APN. Then there is $\lambda\neq 0\in\mathbb{F}^n$ such that the linear space of $F_\lambda$ is trivial.
\end{theorem} 

\begin{proof}
	Since, by Lemma \ref{linear-subspace-splitting}, a B.f. has a nonzero linear structure if and only if its splitting index is $1$, so we simply show that for any APN function $F$ it is impossible to have $\sigma(F_\lambda)=1$, for all $\lambda\neq 0 \in\mathbb{F}^n$.
	
	Suppose, by contradiction, that $F$ is APN and $\sigma(F_\lambda)=1$, for all $\lambda\neq 0 \in\mathbb{F}^n$. By Lemma~\ref{balanced-split-1}, we can suppose that, for any $\lambda\neq 0 \in\mathbb{F}^n$, there are nonzero $v$, $u$ and $w$ not $\Gamma(F_\lambda)$ such that $D_vF_\lambda$ is constant, $D_uF_\lambda$ and $D_wF_\lambda$ are both unbalanced. So we have $\mathcal{F}^2(D_0F_\lambda)=\mathcal{F}^2(D_vF_\lambda)=2^{2n}$, and both $\mathcal{F}^2(D_uF_\lambda)$ and $\mathcal{F}^2(D_wF_\lambda)$ are nonzero positive integers (recall that, for any B.f. $f$, $\mathcal{F}(f)=0$ if and only if $f$ is balanced). Thus, we have \begin{align*}\sum_{a\in\mathbb{F}^n}\mathcal{F}^2(D_aF_\lambda)&\geq\mathcal{F}^2(D_0F_\lambda)+\mathcal{F}^2(D_vF_\lambda)+\mathcal{F}^2(D_uF_\lambda)+\mathcal{F}^2(D_wF_\lambda)\\&=2^{2n}+2^{2n}+\mathcal{F}^2(D_uF_\lambda)+\mathcal{F}^2(D_wF_\lambda)>2^{2n+1}\end{align*} from which we deduce that \[\sum_{\lambda\neq 0\in\mathbb{F}^n}\sum_{a\in\mathbb{F}^n}\mathcal{F}^2(D_aF_\lambda)>2^{2n+1}(2^n-1).\] Thus, by Theorem \ref{APN}, it impossible for $F$ to be an APN function. So it follows that if $F$ is an APN function in even dimension then there is a component whose linear space is trivial. 
\end{proof}

\begin{proposition}[\cite{Cal}]\label{APN-permutation-derivatives}
	Let $F$ be an APN permutation over $\mathbb{F}^n$, with $n$ even. If there are $\lambda\neq 0, a\neq 0\in\mathbb{F}^n$ such that $D_aF_\lambda$ is constant, then $D_aF_\lambda=1$.
\end{proposition}

In the next result we talk about the maximum possible dimension for linear spaces of components of APN permutation.

\begin{theorem}\label{APN-permutation-linear space}
	If $F$ is an APN permutation over $\mathbb{F}^n$, with $n$ even, then \\$\dim V(F_\lambda)\leq 1$, for all $\lambda\neq 0 \in\mathbb{F}^n$. 
\end{theorem}

\begin{proof}
	Suppose, by contradiction, that there is $\mu\neq 0\in\mathbb{F}^n$ such that $\dim V(F_\mu)>1$. It follows that $V(F_\mu)$ contains at least three nonzero elements. Let $a,b\in V(F_\mu)$ be nonzero and distinct. Then, by Proposition~\ref{APN-permutation-derivatives}, we have $D_aF_\lambda=D_bF_\lambda=1$. Clearly, $a+b$ is also a nonzero element in $V(F_\mu)$ different from $a$ and $b$. Note that $D_{a+b}F_\mu(x)=D_aF_\mu(x)+D_bF_\mu(x+a)$, $x\in\mathbb{F}^n$. By Equation \eqref{equivalence-derivatives-eqn} in the proof Lemma~\ref{equivalence-properties}, $D_bF_\mu(x+a)=(D_{Ib}F_\mu)\circ \varphi\sim_A D_bF_\mu$, with $\varphi(x)=Ix+a$ and $I$ being the identity matrix of $GL_n(\mathbb{F})$. Since $D_bF_\mu=1$, so we must have $D_bF_\mu(x+a)=1$. Thus, $D_{a+b}F_\mu(x)=0$, which is impossible by Proposition \ref{APN-permutation-derivatives}. Thus, we must have  $\dim V(F_\lambda)\leq 1$, for all $\lambda\neq 0\in\mathbb{F}^n$.
\end{proof}

\subsection{Quadratic APN functions}
A quadratic v.B.f. from $\mathbb{F}^n$ to itself is denoted by $Q$, the linear space $V(Q_\lambda)$ of a component $Q_\lambda$ is denoted by $V_\lambda$ and we let $V_\lambda^*=V_\lambda\setminus\{0\}$. It is well-known that any APN function cannot contain linear components, so we assume that $Q$ is pure quadratic. For quadratic functions, it is clear from Theorem \ref{quadratic} that we have trivial linear space if and only if the function is bent. So, by Theorem~\ref{APN-components-linear-space}, any quadratic APN functions must have some bent components. In this subsection we are mainly counting how many bent components are in quadratic APN functions.

First we prove a result which relate the dimensions of linear spaces for components of $Q$ to quadratic APN functions. 

\begin{proposition}\label{quadratic-APN-subspace}
	For any quadratic $Q:\mathbb{F}^n\rightarrow\mathbb{F}^n$, we have 
	\begin{align}\label{semi-bent-eqn}
	\sum_{\lambda\neq 0\in\mathbb{F}^n}(2^{\dim V_\lambda}-1)\geq 2^n-1.
	\end{align} 
	Moreover, equality holds if and only if $Q$ is APN. 
\end{proposition}

\begin{proof}
	Since $\mathcal{F}^2(D_0Q_\lambda)=2^{2n}$, so we have
	\begin{align}\label{fourier-eqn-1}\sum_{\lambda\neq 0\in\mathbb{F}^n}\sum_{a\in\mathbb{F}^n}\mathcal{F}^2(D_aQ_\lambda)& = \sum_{\lambda\neq 0\in\mathbb{F}^n}[\mathcal{F}^2(D_0Q_\lambda)+\sum_{a\neq 0\in\mathbb{F}^n}\mathcal{F}^2(D_aQ_\lambda)]\nonumber\\&=\sum_{\lambda\neq 0\in\mathbb{F}^n}[2^{2n}+\sum_{a\neq 0\in\mathbb{F}^n}\mathcal{F}^2(D_aQ_\lambda)]\nonumber\\&=2^{2n}(2^n-1)+\sum_{\lambda\neq 0,\in\mathbb{F}^n}\sum_{a\neq 0\in\mathbb{F}^n}\mathcal{F}^2(D_aQ_\lambda).\end{align}
	By Theorem \ref{APN} and Equation \eqref{fourier-eqn-1}, we deduce that 	
	\begin{align}\label{fourier-eqn-2} \sum_{\lambda\neq 0,\in\mathbb{F}^n}\sum_{a\neq 0\in\mathbb{F}^n}\mathcal{F}^2(D_aQ_\lambda)\geq 2^{2n}(2^n-1)\end{align} and equality holds if and only if $Q$ is APN. 
	
	For any quadratic $Q$, $\deg(D_aQ_\lambda)=0$ if $a\in V_\lambda$ and $\deg(D_aQ_\lambda)=1$ if $a\notin V_\lambda$.  So we have $\mathcal{F}^2(D_aQ_\lambda)=2^{2n}$ if $a\in V_\lambda$ and $\mathcal{F}^2(D_aQ_\lambda)=0$ if $a\notin V_\lambda$.  Thus, we have
	\begin{align}\label{fourier-eqn-3} \sum_{\lambda\neq 0\in\mathbb{F}^n}\sum_{a\neq 0\in\mathbb{F}^n}\mathcal{F}^2(D_aQ_\lambda)&=\sum_{\lambda\neq 0\in\mathbb{F}^n}\sum_{a\in V_\lambda^*}\mathcal{F}^2(D_aQ_\lambda)\nonumber\\&=\sum_{\lambda\neq 0\in\mathbb{F}^n}2^{2n}|V_\lambda^*|\nonumber\\&=2^{2n}\sum_{\lambda\neq 0\in\mathbb{F}^n}(2^{\dim V_\lambda}-1).\end{align} We deduce, from the relation \eqref{fourier-eqn-2} and Equation \eqref{fourier-eqn-3}, that \[\sum_{\lambda\neq 0\in\mathbb{F}^n}(2^{\dim V_\lambda}-1)\geq 2^n-1\] and equality holds if and only if $Q$ is APN.
\end{proof}

It can be easily shown that for any quadratic B.f. $f$ in odd dimension we have $\dim V(f)\geq 1$ and equality holds when $f$ is a semi-bent. This implies that the equality in the relation~\eqref{semi-bent-eqn} happens when $Q$ is an AB function. Since in odd dimension all quadratic functions have non-trivial linear space, then all components of any quadratic APN function in odd dimension have non-trivial linear space. Thus, it implies that the result in Theorem~\ref{APN-components-linear-space} cannot be extended to APN function in odd dimension.

Now we focus on quadratic v.B.f. in even dimension. From Theorem \ref{quadratic}, it is clear that any quadratic B.f. in even dimension has a splitting index $1$ or $2$. By Corollary \ref{quadratic-bent-semi-bent}, we deduce that quadratic is bent if and only if the splitting index is $2$.

\begin{definition}\label{def-quad} For any quadratic $Q$, define \[\Delta_i=\{\lambda\in\mathbb{F}^n\mid \lambda\neq 0, \sigma(Q_\lambda)=i\}, \hspace{.5cm} N=|\Delta_1| \hspace{.5cm}\text{ and }\hspace{.5cm} B=|\Delta_2|.\]
\end{definition} 

\begin{remark}\label{bent-non-bent-remark}
	From Definition \ref{def-quad}, $N$ is the number of non-bent compoments and $B$ is the number of bent components in $Q$ and so we have $N+B=2^n-1$.
\end{remark}

 Nyberg in \cite{Nyb}, proved that bent functions exist only from $\mathbb{F}^n$ to $\mathbb{F}^m$, with $m\leq n/2$, so it well-known that no v.B.f. from $\mathbb{F}^n$ to itself is bent. 
 
 \begin{remark}\label{max-bent}
 The maximum number of bent components in any v.B.f. from $\mathbb{F}^n$ to itself is $2^n-2^{\frac{n}{2}}$ (see \cite{Pott}). So $0\leq B \leq 2^n-2^{\frac{n}{2}}$. In \cite{Mes}, no plateaued APN function has the maximum number of bent components. It is well-known that quadratic functions are plateaued. 
 \end{remark}
 
 In the next result we wish to determine $B$ when $Q$ is an APN function and contains only bent and semi-bent components. 

\begin{proposition}\label{quadratic-APN-bent-semi-bent}
	Let a quadratic $Q:\mathbb{F}^n\rightarrow\mathbb{F}^n$, with $n$ even, be such that $Q_\lambda$, with $\lambda\neq 0$, is bent or semi-bent. Then $Q$ is APN if and only if there are exactly $\frac{2}{3}(2^n-1)$ bent components.
\end{proposition}

\begin{proof}
	For any quadratic APN $Q$, by Theorem \ref{APN-components-linear-space}, we conclude that $B>0$, that is, some components of $Q$ must be bent (as we require that the linear space of some components must be trivial). Since $n$ is even, so $\dim V_\lambda$ is even (see Remark~\ref{dimension-semi-bent}). From Theorem~\ref{quadratic} and Corollary~\ref{quadratic-bent-semi-bent}, we can deduce that $\dim V_\lambda=0$ if and only if $Q_\lambda$ is bent. That is, $\dim V_\lambda\neq 0$ if $\lambda\in\Delta_1$  and $\dim V_\lambda=0$ if $\lambda\in\Delta_2$. For any quadratic APN $Q$, by Proposition~\ref{quadratic-APN-subspace}, we must have 
	
	\begin{align}\label{APN-eqn}
	\sum_{\lambda\neq 0\in\mathbb{F}^n}(2^{\dim V_\lambda}-1)=2^n-1.
	\end{align}
	
	Since $\dim V_\lambda=0$ if $\lambda\in\Delta_2$, then Equation \eqref{APN-eqn} can be reduced to
	
	\begin{align}\label{APN-eqn-1}
	\sum_{\lambda\in\Delta_1}(2^{\dim V_\lambda }-1)=2^n-1.
	\end{align}
	That is, $Q$ is APN if and only if Equation \eqref{APN-eqn-1} holds. 
	
	If $Q$ is such that $Q_\lambda$, with $\lambda\neq 0$, is bent or semi-bent, then $N$ is the number of semi-bent (i.e., $\dim V_\lambda=2$, for any $\lambda\in\Delta_1$). 
	Thus, Equation \eqref{APN-eqn-1} is true if and only if $(2^2-1)|\Delta_1|=3N=2^n-1\iff N=(2^n-1)/3$. Since $N+B=2^n-1$, so $B=2(2^n-1)/3$.
\end{proof}

It follows from Proposition \ref{quadratic-APN-bent-semi-bent} that any quadratic APN function in even dimension with the set \(\{0,\pm 2^{\frac{n}{2}},\pm 2^{\frac{n+2}{2}}\}\) as its Walsh spectrum has $2(2^n-1)/3$ bent components . It is well-known that the Walsh spectrum of any Gold function in even dimension is \(\{0,\pm 2^{\frac{n}{2}},\pm 2^{\frac{n+2}{2}}\}\), so any Gold function has $2(2^n-1)/3$ bent components. 

\begin{theorem}\label{number-bents}
	Let a quadratic $Q:\mathbb{F}^n\rightarrow\mathbb{F}^n$, with $n$ even, be APN. Then \[2(2^n-1)/3\leq B\leq 2^n-2^{n/2}-2,\] where \(B=2(2^n-1)/3+4t\),
	for some integer $t\geq 0$.
\end{theorem}

\begin{proof}
	Suppose that $Q$ is APN. Since the dimension of the linear space of any quadratic in even dimension is even (see Remark~\ref{dimension-semi-bent}), so it follows that for any $Q_\lambda$, with $\lambda\in\Delta_1$, we have $\dim V_\lambda\geq 2$. If, for any $\lambda\in \Delta_1$, $Q_\lambda$ is semi-bent then we are in Proposition \ref{quadratic-APN-bent-semi-bent}, that is, $B=2(2^n-1)/3$. If some components are neither bent nor semi-bent, then we must have $B>2(2^n-1)/3$ for Equation \eqref{APN-eqn-1} to be satisfied. 
	
	If $Q$ has a component $Q_\mu$, with $\mu\in \Delta_1$, which is not semi-bent, then $\dim V_\mu=2k$, for some $k\geq 2$.  So, for Equation \eqref{APN-eqn-1} to be satisfied, the presence of $Q_\mu$ in $Q$ has to increase the number of bent components by \[\frac{2^{2k}-1}{2^2-1}-1=\frac{2^{2k}-4}{3}=4\left(\frac{2^{2k-2}-1}{3}\right)\] which clearly is divisible by $4$. So it follows that \(B=2(2^n-1)/3+4t\), for some integer $t\geq 0$.
	
	By Remark \ref{max-bent}, we have \(B \leq 2^n-2^{n/2}\). Now we show that it not possible to have \(B=2^n-2^{n/2}\). For some $t\geq 0$, we have $B=2(2^n-1)/3+4t=2[(2^n-1)/3+2t] \not\equiv 0\pmod{4}$ since $(2^n-1)/3+2t$ is odd. Thus, \(B\neq 2^n-2^{n/2}\) since $2^n-2^{n/2}\equiv 0 \pmod{4}$. Hence we must have \(B\leq 2^n-2^{n/2}-2\).
\end{proof}

For any quadratic APN $Q$ in demension $4$, by Theorem \ref{number-bents}, we only have one possibility, that is, $B=10$ (this satisfies Proposition \ref{quadratic-APN-bent-semi-bent}). We state this result in following.    

\begin{corollary}\label{quadratic-APN-dim-4}
	A pure quadratic $Q:\mathbb{F}^4\rightarrow \mathbb{F}^4$ is APN if and only if $B=10$. 
\end{corollary}

Not long time ago, only quadratic APN functions with $B=2(2^n-1)/3$ were known. From Proposition \ref{quadratic-APN-bent-semi-bent}, such functions contain only bent and semi-bent components. As noted earlier Gold functions are example of such functions. It had been conjectured that all quadratic APN functions are equivalent to Gold functions (i.e., all quadratic APN functions have the same number of bent components) until Dillon in 2006 gave an example of quadratic APN with different number of bent components and inequivalent to Gold functions. The Dillon's Example:  
\[F(x)=x^3+z^{11}x^5+z^{13}x^9+x^{17}+z^{11}x^{33}+x^{48}\]
is defined over $\mathbb{F}_{2^6}$, with $z$ primitive. 
Using MAGMA, we found that $F$ has $46$ bent components. That is, it is an example of quadratic APN  with  $B=2(2^n-1)/3+4$ (i.e., $t=1$ by Theorem~\ref{number-bents}). Also by computer search, we found the function: 
\[G(x)=x^3+ z^{53}x^{10}+ z^{41}x^{18}+ z^{59}x^{33}+ z^{43}x^{34}+z^{31}x^{48}\] 
over $\mathbb{F}_{2^6}$, with $z$ primitive, which has the same number of bent components as $F$, the Dillon's Example. From Theorem~\ref{number-bents}, we deduce that, in dimension $6$, all the possibilities for the number of bent components in any quadratic APN function are: $42$, $46$, $50$ and $54$. So far we only know the existence of quadratic APN functions with $42$ (Gold functions and others) and $46$ (Dillon's example) bent components but we are uncertain whether those with $50$ and $54$ exists.

In \cite{Yu}, some quadratic APN functions in dimension $8$ with Walsh spectrum\\ $\{-64,-32,-16, 0, 16, 32, 64\}$ (which is different from the Walsh spectrum of Gold function) are found. These functions are further classified in terms of their distribution of Walsh coefficients and two classes are found. One class has 487 functions and the other one has 12 functions. In a class of 487 functions, we considered the function:
\begin{align*}G'(x)&=z^{249}x^{192} + z^{24}x^{160} + z^{210}x^{144} + z^{69}x^{136} + z^{46}x^{132} + z^{164}x^{130} + z^{43}x^{129}\\& + z^{31}x^{96} + z^{30}x^{80} +z^{115}x^{72} + z^{228}x^{68} + z^{16}x^{66} + z^{228}x^{65} + z^{217}x^{48}\\& + z^9x^{40}+ z^{251}x^{36} + z^{151}x^{34} + z^{77}x^{33} +z^{189}x^{24} + z^{109}x^{20} + z^{191}x^{18}\\& + z^{249}x^{17} + z^{175}x^{12} + z^{130}x^{10} + z^{91}x^9 + z^{59}x^6 + z^{60}x^5+z^{121}x^3\end{align*} and by checking with MAGMA, we found that it contains $2(2^8-1)/3+4=174$ bent components (i.e., $t=1$ by Theorem~\ref{number-bents}) and in the other class, we considered the function:
\begin{align*}
G''(x)&=z^{130}x^{192} + z^{160}x^{160} + z^{117}x^{144} + z^{230}x^{136} + z^{228}x^{132} + z^{162}x^{130}\\& + z^{25}x^{129} + z^{79}x^{96} + z^{204}x^{80} + z^{83}x^{72} + z^{159}x^{68} + z^{234}x^{66} + z^{36}x^{65}\\& + z^{67}x^{48} + z^{151}x^{40} + z^{17}x^{36} + z^{81}x^{34} + z^{52}x^{33} +z^9x^{24} + z^{116}x^{20}\\& + z^{102}x^{18} + z^{97}x^{17} + z^{74}x^{12} + z^{48}x^{10} + z^{144}x^9 + z^{58}x^6 + z^{146}x^5 +z^{123}x^3
\end{align*} which was found to have $2(2^8-1)/3+8=178$ bent components (i.e., $t=2$ by Theorem~\ref{number-bents}). Thus, in dimension $8$, we only know the existence of quadratic APN functions with $170$, $174$ and $178$ bent components and it is yet to be known whether quadratic APN functions having $B=2(2^8-1)/3+4t$, with $3\leq t\leq 17$, bent components exist.

\begin{proposition}\label{quadratic-APN-nonlinearity}
	Let $Q:\mathbb{F}^n\rightarrow\mathbb{F}^n$ be APN with $B=2(2^n-1)/3+4t$, for some positive integer $t$, as described in Theorem~\ref{number-bents}. Then 
	\[\mathcal{N}(Q)=\begin{cases}
	2^{n-1}-2^{n/2} & \hspace{.5cm} \text{ if }t=0, n\geq 4\\
	2^{n-1}-2^{n/2+1} & \hspace{.5cm} \text{ if } 1\leq t\leq 4, n\geq 6
	\end{cases}\]
\end{proposition}

\begin{proof}
	We first need to recall from Remark \ref{dimension-semi-bent}, that for any quadratic B.f. on $n$ variables, with even $n$, the dimension $k$ of its linear space is even and the Walsh spectrum is $\{0,2^{(n+k)/2}\}$. 
	
	If $t=0$ then, by Proposition \ref{quadratic-APN-bent-semi-bent}, all components of $Q$ are bent and semi-bent, that is, the Walsh spectrum of $Q$ is $\{0,\pm 2^{(n+2)/2},\pm 2^{n/2}\}$. So clearly, by Corollary \ref{quadratic-nonlinearity}, $\mathcal{N}(Q)=2^{n-1}-2^{n/2}$.
	
	To prove that $\mathcal{N}(Q)=2^{n-1}-2^{n/2+1}$ if $1\leq t\leq 4$, we need to show that for this range of $t$ we have $\dim V_\lambda\in\{0,2,4\}$, for all $\lambda\neq 0\in\mathbb{F}^n$, that is, Walsh spectrum of $Q$ is $\{0,\pm 2^{(n+4)/2},\pm 2^{(n+2)/2},\pm 2^{n/2}\}$.
	
	It is clear from Theorem \ref{number-bents} that for $t\geq 1$, we have $B>2(2^n-1)/3$, and so Proposition \ref{quadratic-APN-bent-semi-bent} allows us to conclude that there must be $\lambda\neq 0\in\mathbb{F}^n$ such that $\dim V_\lambda>2$. We claim that if $1\leq t\leq 4$, then we have $\dim V_\lambda\in\{0,2,4\}$, for $\lambda\neq 0\in\mathbb{F}^n$. Suppose, by contradiction, that there is $\mu\neq 0\in\mathbb{F}^n$ such that $\dim V_\mu=6$. Then, as noted in the proof of Theorem~\ref{number-bents}, the presence of $Q_\mu$ increases the number of bent components by \[4\left(\frac{2^{6-2}-1}{3}\right)=4(5),\] implying that $B\geq 2(2^n-1)/3+4(5)$. So it follows that if, for some $\lambda\neq 0\in\mathbb{F}^n$, $\dim V_\lambda =6$, then we have $t\geq 5$. This implies that, if $1\leq t\leq 4$, then we must have $\dim V_\lambda\in\{0,2,4\}$, for all $\lambda\neq 0\in\mathbb{F}^n$. So in this case the Walsh spectrum of $Q$ is $\{0,\pm 2^{(n+4)/2},\pm 2^{(n+2)/2},\pm 2^{n/2}\}$ from which we deduce that $\mathcal{N}(Q)=2^{n-1}-2^{n/2+1}$. 
\end{proof}

From Proposition \ref{quadratic-APN-nonlinearity}, it seems like the nonlinearity of any quadratic APN function decreases as the number of bent components increases and it is the highest when the number of bent components is at the lowest possible.

\section{Quadratic power functions} \label{power functions-section}
 Pott et al. in  \cite{Pott} say that the question to determine all monomial bent functions $Tr(\alpha x^d)$ on $\mathbb{F}_{2^n}$, with $\alpha\in\mathbb{F}_{2^n}^*$ and $n$ even, has attracted quite a lot of research interest. In this section we study the Walsh spectrum and enumerate bent components for any quadratic power functions. Recall that a function $F=x^d$ is a quadratic power functions if $d=2^j+2^i$, with $i\neq j$ and $j>i\geq 0$. It is well-known that a function with the power $d=2^i(2^{j-i}+1)$ is affine equivalent to the one with power $d'=2^{j-i}+1$. So we simply consider the power $2^k+1$, for some positive integer $k$. For a function $F$, we denote its image by $\mathrm{Im}(F)$. 
 
 We denote the greatest common divisor integers $m$ and $m'$ by $(m,m')$.  We begin with the following well-known result which can be found in \cite{Eric}.

 \begin{lemma}\label{gcd}
 	For any positive integers $n$ and $k$, we have
 	\begin{itemize}
 		\item[(a)] $(2^n-1,2^k-1)=2^{(n,k)}-1$,
 		\item[(b)] $(2^n-1, 2^k+1)=\begin{cases}1 \text{ if } n/(n,k) \text{ is odd},\\2^{(n,k)}+1 \text{ if } n/(n,k) \text{ is even}.
 		\end{cases}$
 	\end{itemize}
 \end{lemma}

 \begin{theorem}\label{number-power-bent}
 	Let $F(x)=x^{2^k+1}$ be a function in $\mathbb{F}_{2^n}[x]$, with even $n$ and some integer $k\geq 1$. Let $m=(n,k)$, $s=(n,2k)$ and $e=1$ if $n/m$ is odd and $e=2^m+1$ if $n/m$ is even. Then
 	\begin{itemize}
 		\item[(a)] $F$ is an $e$-to-$1$ function, 
 		\item[(b)] $F_\alpha$ is bent if and only if $\alpha\notin \mathrm{Im}(F)$.
 		\item[(c)] the number of bent components for $F$ is \(\frac{(e-1)(2^n-1)}{e}\), 
 		\item[(d)] the Walsh spectrum of $F$ is \(\{0,\pm 2^{(n+s)/2}\}\) if $e=1$, and \(\{0,\pm 2^{(n+s)/2},\pm 2^{n/2}\}\) if $e=2^m+1$.
 	\end{itemize}
 \end{theorem}
 
 \begin{proof} Let $S=\mathrm{Im}(F)\setminus\{0\}=\{\xi^{2^k+1}\mid \xi\in\mathbb{F}_{2^n}^*\}$. It can be easily shown that $S$ is a multiplicative subgroup of $\mathbb{F}_{2^n}^*$.
 	\begin{itemize}
 		\item[(a)] Clearly, $F$ maps $\mathbb{F}_{2^n}^*$ onto $S$. So we only need to show that $S$ has the order $(2^n-1)/e$. Now we need to find the order of $S$.  First observe that every element $\zeta$ in $S$ satisfies $\zeta^{(2^n-1)/e}=1$, where $e=(2^n-1,2^k+1)$. By Lemma~\ref{gcd}, $e=1$ if $n/m$ is odd and $e=2^m+1$ if $n/m$ is even. If $\nu$ is a primitive element in $\mathbb{F}_{2^n}$, then the order of $\nu^{2^k+1}$ is $\mathrm{ord}(\nu^{2^k+1})=\mathrm{ord}(\nu^e)=(2^n-1)/e$. Clearly, $\nu^{2^k+1}$ has the highest order in $S$. It is well-known that $\mathbb{F}_{2^n}^*$ is cyclic group, so $S$ being a subgroup must be cyclic with $\nu^{2^k+1}$ as a generator.  Thus, it follows that the order of $S$ is $(2^n-1)/e$, implying that $F$ is an $e$-to-$1$ function. 
 		
 		\item[(b)] It is equivalent to show that $F_\alpha$ is non-bent if and only if $\alpha\in \mathrm{Im}(F)$. $F_\alpha$ is bent if its linear space is trivial, so we need to prove that the dimension of the linear space of $F_\alpha$ is non-trivial, that is,  $\dim V_\alpha \geq 1$ if and only if $\alpha\in \mathrm{Im}(F)$. 
 		
 		A component $F_\alpha$, with $\alpha\in \mathbb{F}_{2^n}$, is non-bent if there exists $\beta$ in $\mathbb{F}_{2^n}^*$ such that $D_\beta F_\alpha$ is constant. Suppose that $F_\alpha$, with $\alpha\in\mathbb{F}_{2^n}^*$, is non-bent and $D_\beta F_\alpha$ is constant, with $\beta\in\mathbb{F}_{2^n}$. So we have
 		
 		\begin{align}\label{derivative-compoment-eq}
 		D_\beta F_\alpha(x)&=F_\alpha(x)+F_\alpha(x+\beta)=Tr\left(\alpha x^{2^k+1}\right)+Tr\left(\alpha(x+\beta)^{2^k+1}\right)\nonumber\\&=Tr\left(\alpha x^{2^k+1}\right)+Tr\left(\alpha(x^{2^k}+\beta^{2^k})(x+\beta)\right)\nonumber\\&=Tr\left(\alpha x^{2^k+1}\right)+Tr\left(\alpha(x^{2^k+1}+\beta x^{2^k}+\beta^{2^k}x+\beta^{2^k+1})\right)\nonumber\\&=Tr\left(\alpha\beta x^{2^k}\right)+Tr\left(\alpha\beta^{2^k}x\right)+Tr\left(\alpha \beta^{2^k+1}\right)\nonumber\\&=Tr\left((\alpha\beta+\alpha^{2^k}\beta^{2^{2k}})x^{2^k}\right)+Tr\left(\alpha \beta^{2^k+1}\right).
 		\end{align}
 		Observe that $D_\beta F_\alpha$ is constant if and only if, in Equation \eqref{derivative-compoment-eq}, we have \[Tr\left((\alpha\beta+\alpha^{2^k}\beta^{2^{2k}})x^{2^k}\right)=0.\] This happens if and only if \[\alpha\beta+\alpha^{2^k}\beta^{2^{2k}}=\alpha\beta\left(1+\alpha^{2^k-1}\beta^{2^{2k}-1}\right)=0.\]  So either $\beta=0$ or \begin{align}\label{eqn}\alpha^{2^k-1}\beta^{2^{2k}-1}=(\alpha\beta^\ell)^{2^k-1}=1,\end{align} with $\ell=\frac{2^{2k}-1}{2^k-1}=2^k+1$. Suppose that $\zeta$ is a primitive element in $\mathbb{F}_{2^n}$. Then we can write $\alpha=\zeta^r$ and $\beta=\zeta^t$, for some integers $r$ and $t$. So it follows that Equation \eqref{eqn} becomes $\zeta^{(r+t\ell)(2^k-1)}=1$ which implies that \[(r+t\ell)(2^k-1)=r(2^k-1)+t(2^{2k}-1)=c(2^n-1),\] for some integer $c$. Thus, we have\\
 		
 		\(r=\frac{c(2^n-1)}{2^k-1}-\frac{t(2^{2k}-1)}{2^k-1}=\frac{c(2^n-1)}{2^k-1}-t(2^k+1)=e\left(\frac{c(2^n-1)}{e(2^k-1)}-\frac{t(2^k+1)}{e}\right).\)\\
 		
 		Recall that $e=(2^n-1,2^k+1)$.   So all $\alpha$'s which satisfy $(\alpha\beta^\ell)^{2^k-1}=1$ must be those which satisfy $\alpha^{(2^n-1)/e}=1$. These are elements whose orders are divisors of $(2^n-1)/e$. It implies that $\alpha\in S$. Including $\alpha=0$, it follows that $F_\alpha$ has a non-trivial linear space if and only if $\alpha\in\mathrm{Im}(F)$.
 		
 		\item[(c)] By part~(b), we deduce that the number of bent components is $2^n-|\mathrm{Im}(F)|$. Since $|\mathrm{Im}(F)|=1+|S|=1+(2^n-1)/e$, then the number of bent components is \[2^n-|\mathrm{Im}(F)|=\frac{(e-1)(2^n-1)}{e}.\] 
 		
 		\item[(d)] We first determine $V_\alpha$, for any $\alpha\in\mathbb{F}_{2^n}^*$, and then use Theorem \ref{quadratic-nonlinearity} to deduce the Walsh spectrum of $F$. In part~(b), we showed that $V_\alpha=\{0\}$ if $\alpha\notin \mathrm{Im}(F)$ (i.e., $F_\alpha$ is bent) and $|V_\alpha|>1$ if $\alpha\in \mathrm{Im}(F)$. For any  $\alpha\in S=\mathrm{Im}(F)\setminus\{0\}$, we also showed, in part~(b), that $D_\beta F_\alpha$ is constant if either $\beta$ is equal to $0$ or satisfies $(\alpha\beta^{2^k+1})^{2^k-1}=1$. Thus, we have $\beta^{2^{2k}-1}=(\alpha^{-1})^{2^k-1}$. If $\alpha =1$, then $\beta\in \mathbb{F}_{2^s}^*$, with $s=(2k,n)$ and if $\alpha\neq 1$, then $\beta \in \mu\mathbb{F}_{2^s}^*$, where $\mu$ is $\ell$-th root of $\alpha^{-1}$.  So it follows that $|V_\alpha|=2^s$.
 		
 		Given that $m=(n,k)$, by Lemma \ref{gcd}, we have $e=1$ if $n/m$ is odd and $e=2^m+1$ if $n/m$ is even. If $e=1$ then, by part (a), $F$ is a permutation which implies that it has no bent components and so we have $|V_\alpha|=2^s$, for all $\alpha\in\mathbb{F}_{2^n}^*$. This implies that its Walsh spectrum of $F$ is $\{0,\pm 2^{(n+s)/2}\}$ (see Theorem~\ref{quadratic-nonlinearity}). If $e=2^m+1$, then $F$ contains bent components and as shown above, all the linear spaces of non-bent components have the same order $2^s$, implying that the Walsh spectrum of $F$ is $\{0,\pm 2^{(n+s)/2},\pm 2^{n/2}\}$.
 	\end{itemize}  
 \end{proof}

 \begin{corollary}\label{power-quadratic-APN}
 	Let $F(x)=x^{2^k+1}$ be a power polynomial in $\mathbb{F}_{2^n}[x]$, with positive integers $n$ and $k\geq 1$ and let $e=(2^n-1,2^k+1)$ and $s=(n,2k)$. Then $F$ is APN if and only if $e=3$ and $s=2$. 
 	Equivalently, $F$ is APN if and only if there are exactly $2(2^n-1)/3$ bent components and the rest semi-bent.
 \end{corollary}
 
 \begin{proof}
 	By Theorem \ref{number-power-bent}, there are $(2^n-1)/e$ (non-trivial) non-bent components for $F$ and their linear spaces have the same order $2^s$. Since $n$ is even then $s=2t$, where $t=(k,n/2)$.  Thus, by Proposition \ref{quadratic-APN-subspace}, $F$ is APN if and only if \begin{align}\label{power-quadratic-APN-eqn}\left(\frac{2^n-1}{e}\right)(2^s-1)=2^n-1.\end{align}
 	Since Equation \eqref{power-quadratic-APN-eqn} holds if and only if $e=2^s-1$, then we conclude that $(2^s-1)|(2^k+1)$. Since $t|s$ then $(2^t-1)|(2^s-1)$, implying that $(2^t-1)|(2^k+1)$. But also $(2^t-1)|(2^k-1)$ (recall that $t|k$), so it implies that we must have $t=1$ as clearly $2^k-1$ and $2^k+1$ are  relatively prime. Observe that $t=1$ implies $s=2$, so it follows that $F$ is APN if and only if $s=2$ and $e=2^s-1=3$. In other words, $F$ is APN if and only if the number of bent components is exactly \(2(2^n-1)/3\) and the other components are semi-bent. 
 \end{proof}
 
 From Theorem \ref{number-power-bent}, we observe that a quadratic power function has some bent components if $e\geq 3$ and equality gives the lowest number of bent components we can get and also when $F$ is APN. So we state this in the following.
 
 \begin{corollary}
 	If a quadratic power function, in even dimension, has some bent components, then they are at least $2(2^n-1)/3$.
 \end{corollary}

\section*{Acknowledgement} The results in this paper appear in the first author's PhD thesis supervised by the second author.

\end{document}